\documentclass[
submission
, nomarks
]{dmtcs-episciences}

\usepackage[utf8]{inputenc}
\usepackage{subfigure}
\usepackage{enumerate}

\usepackage{amsmath,amssymb,amsfonts,amstext}
\usepackage{bbm}
\usepackage{enumitem,mdwlist} 
\usepackage{cases} 
\usepackage{varwidth} 
\usepackage{subfigure} 
\usepackage{mathrsfs} 
\usepackage{complexity}
\usepackage{indentfirst}

\newtheorem{theorem}{Theorem}[section]

\newtheorem{lemma}[theorem]{Lemma}

\usepackage[round]{natbib}

\author{Diane Castonguay
  \and Elisângela S. Dias
  \and Fernanda N. Mesquita
  \and Julliano R. Nascimento} 
\title[Computing a $3$-role assignment is polynomial-time solvable on complementary prisms]{Computing a $3$-role assignment is polynomial-time solvable on complementary prisms\thanks{This study was financed in part by the Fundação de Amparo à Pesquisa do Estado de Goiás (FAPEG)}}
\affiliation{
  Instituto de Inform\'{a}tica, Universidade Federal de Goi\'{a}s, GO, Brazil}

\received{ }
\revised{ }
\accepted{ }

\begin{document}
\publicationdetails{VOL}{2022}{ISS}{NUM}{SUBM}
\maketitle
\begin{abstract}
 A $r$-\textit{role assignment} of a simple graph $G$ is an assignment of $r$ distinct roles to the vertices of $G$, such that two vertices with the same role have the same set of roles assigned to related vertices. Furthermore, a specific $r$-role assignment defines a \textit{role graph}, in which the vertices are the distinct $r$ roles, and there is an edge between two roles whenever there are two related vertices in the graph $G$ that correspond to these roles. We consider complementary prisms, which are graphs formed from the disjoint union of the graph with its respective complement, adding the edges of a perfect matching between their corresponding vertices. In this work, we characterize the complementary prisms that do not admit a $3$-role assignment. We highlight that all of them are complementary prisms of disconnected bipartite graphs. Moreover, using our findings, we show that the problem of deciding whether a complementary prism has a $3$-role assignment can be solved in polynomial time.\end{abstract}
 
\section{Introduction}

In $1980$, Angluin introduced the concept of \textit{covering}, from which role assignment arises, as a tool for networks of processors~\citep{angluin1980local}. A decade later, based on graph models for social networks, formalized the notion of role assignment under the name of \textit{role coloring}~\citep{everett1991role}. Indeed, a $r$-\textit{role assignment} of a simple graph $G$ is an assignment of $r$ distinct roles to the vertices of $G$, such that if two vertices have the same role, then the sets of roles of their neighbors coincide. Moreover, for such an assignment, we obtain a \textit{role graph}, where the vertices are the $r$ distinct roles and there is an edge between two roles whenever there are two neighbors in the graph $G$ that correspond to these roles. Note that, the role graph has not multiple edges, but allows loops since two related vertices in $G$ can have the same role. Observe that, while a social network usually give rise to a large graph, a role assignment allows to represent the same network through a smaller graph.

We defined the $r$-\textsc{Role Assignment} problem as follows:
\begin{center}
	\fbox{\begin{varwidth}{\dimexpr\textwidth-2\fboxsep-2\fboxrule\relax} 
        $r$-\textsc{Role Assignment} \\
		\textbf{Instance:} A simple graph $G$.
		
       \textbf{Question:} Does $G$ admit a $r$-role assignment?
    \end{varwidth}}
\end{center}

Applications of role assignment are highlighted in several contexts such as social networks, see~\cite{everett1991role, roberts2001hard} and distributed computing, see~\cite{chalopin2006local}. 

In $2001$,~\citeauthor{roberts2001hard} proved the $\NP$-completeness of $2$-\textsc{Role Assignment}. Such a result was expanded in $2005$, by~\citeauthor{fiala2005complete}, who showed that $r$-\textsc{Role Assignment} is $\NP$-complete for any fixed $r\geq 3$. In this context, \cite{purcell2015complexity} showed that $r$-\textsc{Role Assignment}, with $r\geq 2$, complexity decision remains $\NP$-complete for planar graphs, while cographs always have $r$-role assignment, which makes the problem constant for this class.
Considering chordal and split graphs, a dichotomy for the complexity of $r$-\textsc{Role Assignment} arises. For chordal graphs,~\cite{van2010computing} checked that the problem is solvable in polynomial time for $r=2$ and $\NP$-complete for $r\geq 3$. On the other hand, for split graphs,~\cite{dourado2016computing} concluded that problem is trivial, with true answer, for $r= 2$, solvable in polynomial time for $r=3$ and $\NP$-complete for any fixed $r\geq 4$. 


The complementary prism is linked to the notion of complementary product, introduced in the literature in $2009$, by ~\citeauthor{haynes2009domination} as a generalization of the Cartesian product. The authors give a special attention to the particular case of the operation, called \textit{complementary prism} of a graph, which can be seen as a variant of \textit{prism}, which is the Cartesian product of a graph with $K_2$. The \textit{complementary prism} of $G$, denoted by $G \overline{G}$, is the graph formed from the disjoint union of $G$ and its complement $\overline{G}$, adding the edges of a perfect matching between the corresponding vertices of $G$ and $\overline{G}$. \cite{castonguay2018prismas}, characterized complementary prisms that admit a $2$-role assignment, showing that only the complementary prisms of a path with three vertices does not. In this sense,~\cite{castonguay2023prismas}, considered the role graph $K'_{1,r}$ which is the bipartite graph $K_{1,r}$ with a loop at the vertex of degree $r$ and showed that the problem of deciding whether a prism complement has a $(r+1)$-role assignment, when the role graph is $K'_{1,r}$, is $\NP$-complete and set the conjecture that, for $r\geq 3$, $(r+1)$-\textsc{Role Assignment} for complementary prisms is $\NP$-complete.

In this paper, we characterize complementary prisms that admit a $3$-role assignment. We point out that the complementary prisms without a $3$-role assignment, arise from disconnected bipartite graphs. Finally, we exhibit a solution in polynomial time for $3$-\textsc{Role Assignment} for complementary prisms.

This paper is organized as follows. In Section~\ref{section_preliminary}, we set notations and terminology. Complementary prisms without a $3$-role assignment are characterized in Section~\ref{section_5.7}. Ours general results for complementary prisms with $3$-role assignment are presented in Section~\ref{section_5.3}. Then, specific results for bipartite graphs in Section~\ref{section_5.6} with $3$-role assignment and of non-bipartite graphs in Section~\ref{section_5.8}. Finally, in Section~\ref{subsection_characterization_prism} have conditions to have a $3$-role assignment in complementary prisms.

\section{Preliminaries} 
\label{section_preliminary}
A \textit{graph} $G$ is  a pair $(V(G), E(G))$, where $V(G)$ is the set of vertices and $E(G)$ is the set of edges. The vertices $u$ and $v$ are \textit{adjacent} or \textit{neighbors} if they are joined by an edge $e$, also denoted by $uv$.  
A \textit{loop}  is an edge incident to only one vertex. The \textit{neighborhood} of a vertex $v$, denoted by $N_{G}(v)$, is the set of all neighbors of $v$ in $G$. 
A \textit{simple graph} is a graph without loops.  In a simple graph $G$, the \textit{degree} of a vertex $v$ is the cardinality of $N_{G}(v)$. The neighborhood of a subset $U$ of $V(G)$, denoted as $N_{G}(U)$, is the union of the neighborhoods of the vertices of $U$. A vertex of degree zero is \textit{isolated} and a \textit{leaf} is a vertex of degree $1$. We denote by $G_i$ the set of vertices of $G$ with degree $i$. With that notation, $G_0$ is the set of isolated vertices of $G$ and $G_1$ the set of leaves of $G$. A \textit{universal vertex}  is a vertex adjacent to all other vertices of the graph.

A \textit{path} is a sequence of distinct vertices with an edge between each pair of consecutive vertices. For $n \geq 2$, we denote a path on $n$ vertices by $P_n$ or by the sequence of vertices $v_1 \ldots v_n$.


A graph $G$ is a \textit{bipartite graph} if one can partition $V(G) = A \cup B$ so that if there is an edge $uv \in E(G)$, then $u \in A$ and $v \in B$, or vice versa. In this case, we say that $(A, B)$ is a \textit{bipartite partition} of $G$.
A \textit{clique} is a subset of vertices that are pairwise adjacent. A \textit{complete graph}, denoted by $K_n$, is a graph of $n$ vertices which is $V(K_n)$ a clique. 

Given a simple graph $G$ and a graph $R$, possibly with loops. A \textit{$R$-role assignment} of $G$ is a surjective vertex mapping $p: V(G) \rightarrow V(R)$ such that $p(N_{G}(v))=N_{R}(p(v))$ for all $v \in V(G)$. A graph $G$ has a \textit{$r$-role assignment} if it admits a $R$-role assignment for some graph $R$, called the \textit{role graph}, with $|V(R)|=r$. We set $1, \ldots , r$ the vertices of $R$, also called \textit{roles}. From now on, all graphs (except maybe the role graph) are simple. Observe that, if the graph $G$ is connected, then the role graph $R$ of any role assignment of $G$ is also connected. Also, if role graph $R$ is bipartite, then so is $G$. 

In Figure~\ref{figure_3_attribution}, we list the possible role graphs with three vertices (up to isomorphism). We note that all role graphs have $1,3 \in N_R(2)$ and if $1 \in N_R(1)$, then $3 \in N_R(3)$.  In all figures, we standardize that the black vertices receive role~1, white vertices, role~2, and, finally, gray vertices, role~3.

\begin{figure} [ht!]
	\centering
	\includegraphics[scale=0.85]{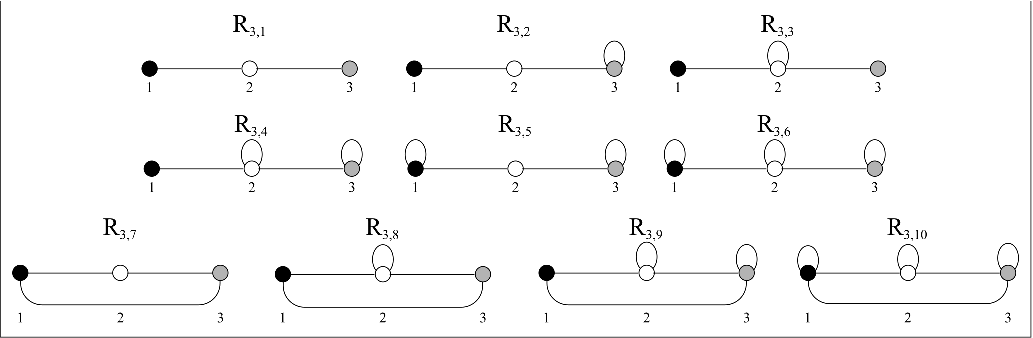} 
	\caption{Possible role graphs arising from a $3$-role assignment.}
	\label{figure_3_attribution}
\end{figure}

The \textit{complementary prism} \index{complementary prism} of $G$, denoted by $G \overline{G}$, is the graph formed from the disjoint union of the graph $G$ and its complement $\overline{G}$, adding the edges of a perfect match between the corresponding vertices of $G$ and $\overline{G}$.
 Furthermore, for a vertex $v$ of $G$, we denote by $\overline{v}$ the corresponding vertex in $\overline{G}$ and, for a set $X \subseteq V (G)$, $ \overline{X}$ denote the corresponding set of vertices in $V (\overline{G})$. Thus, $V(G\overline{G})=V(G) \cup \overline{V(G)}$ and $E(G\overline{G})=E(G)\cup E(\overline{G})\cup \{v\overline{v} \mid v\in V(G)\}$.

Figure~\ref{figure_C5_C5} presents the complementary prism of $C_5$, which is the Petersen graph, together with a $R_{3,2}$-role assignment. For the sake of clarity, we omit to draw some edges. In this sense, three spaced lines between the highlighted subgraphs $C_5$ and $\overline{C_5}$ designate the perfect matching between these subgraphs.  

\begin{figure} [ht!]
	\centering
	\includegraphics[scale=0.8]{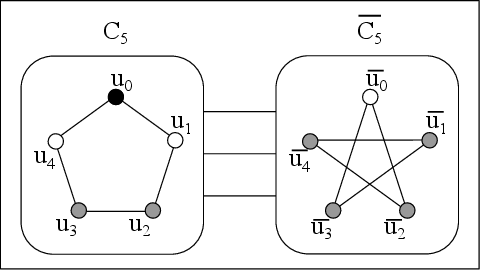} 
	\caption{A $R_{3,2}$-role assignment from the complementary prism of $C_5$.}
	\label{figure_C5_C5}
\end{figure}

For $t \geq 2$, denote by $K^t_2$, the union of $t$ copies of $K_2$. Following the same draw pattern, Figure~\ref{grafo_K_2^3} shows the complementary prism of  $K^3_2$ (with a $R_{3,7}$-role assignment). Again, we omit to draw some edges and represent by a double line between highlighted subgraphs all possible edges between them. 
 
\begin{figure}[ht!]
	\centering
	\includegraphics[scale=0.75]{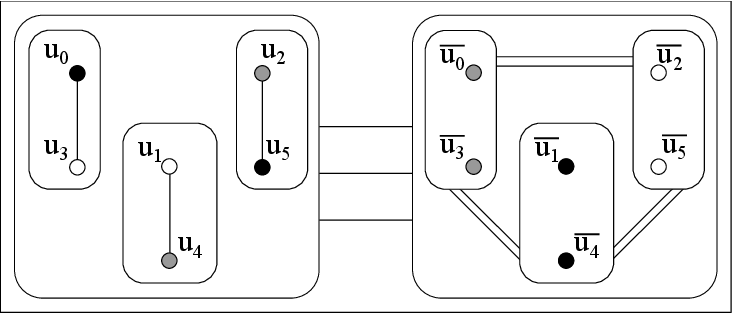} 
	\caption{A $R_{3,7}$-role assignment for the complementary prism of $K_2^3$.}
	\label{grafo_K_2^3}
\end{figure} 

\section{Complementary prisms without a 3-role assignment}
\label{section_5.7}

We present the complementary prism graphs with no possible $3$-role assignment. First, we present the case of the complementary prism of a complete graph.

\begin{lemma} \label{lemma:prisma_Kn_nao_tem} The complementary prism of $K_n$, $n\geq 2$ does not have a $3$-role assignment.\end{lemma}

\begin{proof} By contradiction, we assume that $p: V(K_n\overline{K_n}) \rightarrow V(R)$ is a $3$-role assignment. For all $x \in V(K_n)$, we have that $|N(\overline{x})|=1$ and $R \in\{R_{3,1}, R_{3,2}, R_ {3,3}, R_{3,4}\}$. In all four possible role graphs $R$, we have that $N_R(1)= \{2 \}$ and $1,3 \in N_R(2)$. So, for all $x \in V(K_n)$, $p(\overline{x})\in \{1,3\}$ and $p(x)=2$. Since $p$ is a $3$-role assignment, there are $u,v \in V(K_n)$, such that $p(\overline{u})=1$ and $p(\overline{v}) =3$. This leads to a contradiction, since $1 \not \in p(N(v))=N_{R}(p(v))=N_{R}(2)$. Thus, the complementary prism of $K_n$ does not have a $3$-role assignment. \end{proof}

In the following lemma, we approach the case of complementary prisms of graphs with isolated vertices. We consider the union of complete bipartite graphs with at least one isolated vertex, except $K_1 \cup K_{1,m}$. In the next section (Lemma~\ref{lemma:k1m,k1}), we show that the complementary prism of $K_1\cup K_{1,m}$, with $m \geq 1$, has a $3$-role assignment. Recall that $G_0$ is the set of isolated vertices of a graph $G$.

\begin{lemma} \label{lemma:k_n,mUk_1_nao_has_3attribution} Let $G$ be a bipartite graph with isolated vertices, such that $G\not \simeq K_1\cup K_{1,m}$ with $m \geq 1$. If every non-trivial connected component of $G$ is isomorphic to $K_{n,m}$ for some $n,m \geq 1$, then the complementary prism of $G$ does not have a $3$-role assignment. \end{lemma}

\begin{proof} By contradiction, we assume that $p: V(G\overline{G}) \rightarrow V(R)$ is a $3$-role assignment. We consider $a \in G_0$. We can assume, without loss of generality, that $p(a)=1$. We know that $1,3\in N_R(2)$, $p(\overline{a})=2$ and $N_R(1)= \{2 \}$. Since $|N_{R}(2)|\geq 2$, $2\not \in p(G_0)$ and $p(\overline{G} _0)=\{2\}$. On the other hand, there is $u \in V(G)-G_0$, such that, $p(\overline{u})= 3$. We have two cases to check: $p(u)=2$ and $p(u)=3$.

If $p(u)=2$, then there is $v \in N_{G}(u)$, such that $p(v)=1$. Therefore, $p(N_G(v))=\{2\}$ and $p(\overline{v})=2$. Given that $av \not \in E(G)$, we have $2\in N_R(2)$. Therefore, there is $w \in N_G(u)$, such that $p(w)=2$. By hypothesis, $N(w)= \{ \overline{w}\} \cup N_G(v)$ and $p(N(w))=\{p(\overline{w}),2\},$ a contradiction, since $|N_R(2)|=3$.

If $p(u)=3$, then there is $v \in N_G(u)$, such that $p(v)=2$. By the previous case, $p(\overline{v}) \neq 3$. We will analyze the possible roles for $p(\overline{v})\in \{1, 2\}$. In the case that $p(\overline{v})=1$, we have that $p(N(\overline{v}))= \{2 \}$. As $G \not \simeq K_1 \cup K_{1,m}$, with $m \geq 1$, we have that $N(\overline{v})-(\{\overline{a},v\} ) \neq \emptyset$, so $2\in N_R(2)$. So there is $w \in N_G(v)$, such that $p(w)=2$. 
Observe that $N(w)= \{ \overline{w}\} \cup N_G(u)$. Since $p(u)=3$ and $N_R(1)= \{2 \}$, we get that $1 \not \in p(N_G(u))$. However, because $p(\overline{u})=3$ and $wu \not \in E(G)$, we have that $p(\overline{u}) \neq 1$. This leads to a contradiction, as $1 \not \in p(N(w))$ and $p(w)=2$.
Thus, $p(\overline{v})=2$ and $N_R(2)=\{1,2,3\}$. So there exists $w \in V(G)-N_{G}[v]$, such that $p(\overline{w})=1$. Therefore, $p(N(\overline{w}))=\{2\}$. We note that $w \in N_{G}(u)$, since $1\not \in N_{R}(3)$ and $p(\overline{u})=3$. By the hypothesis, $N_{G}(v)=N_{G}(w)$. We have that $N(\overline{v})=\{v\} \cup (\overline{V(G)-N_{G}[v]})= \{v\}\cup (\overline{V (G)-(N_{G}(w)\cup \{v\})})\subseteq \{v,\overline{w}\} \cup N(\overline{w})$, hence $p (N(\overline{v}))=\{1,2\}$, a contradiction.\end{proof}

We present below the case of complementary prisms of graphs without isolated vertices. Note that the graph $K_2$ is included in Lemma~\ref{lemma:prisma_Kn_nao_tem}. In the case of the graph composed of two or three copies of $K_2$, the complementary prism has a $3$-role assignment, as we will see in Lemmas~\ref{lemma:k1m,k1} and Figure~\ref{grafo_K_2^3}. So, the first case is the union of at least four copies of $K_2$. Recall that $G^t$ denotes $t$  copies of the graph $G$.

\begin{lemma} \label{lemma:split_(K_2)^t} The complementary prisms of $K_2^t$, with $t \geq 4$, does not have a $3$-role assignment. \end{lemma}

\begin{proof} Let $G \simeq K^t_2$, with $t\geq 4$. By contradiction, we assume that $p: V(G\overline{G}) \rightarrow V(R)$ is a $3$-role assignment. As all vertices of $G$ are leaves, if we have a loop on the role~2, then we have that for all $x\in V(G)$, $p(x)\in \{1,3\}$ and $p(\overline{x})=2$. So $p(N(\overline{x}))=\{p(x),2\},$ is a contradiction. Therefore, $N_{R}(2)=\{1,3\}$. 
Since $G\overline{G}$ is non-bipartite, $R\neq R_{3,1}$.
As $K_4 \subseteq \overline{G}$, we observe that $R \neq R_{3,7}$ since at least one role would be repeated in $K_4$. Therefore, $R=R_{3,2}$ or $R_{3,5}$. So $1 \not \in N_{R}(3)$ and there is a loop on the role~3.

We denote by $u_iv_i$, for $i=1, \ldots, t$, the edges of $G$. Suppose initially that $1\in p(V(G))$. We can assume, without loss of generality, that $p(u_1)=1$. Next, we will analyze the possible roles for $\overline{u_1}$. Remember that $p(\overline{u}_1)\neq 3$.

If $p(\overline{u_1})=1$, then $p(v_1)=2$ and $p(\overline{v_1})=3$. Therefore, $N_R(1)=\{1,2\}$. As $1\not \in N_{R}(3)$, we have $p(\overline{u_2})=2$ and $p(\overline{u_3})=2$. So, there is a loop on the role~2, a contradiction.

If $p(\overline{u_1})=2$, then we can assume that $p(\overline{u_2})=3$. So, $p(\overline{u_3}), p(\overline{v_3}) \in N_{R}(2)\cap N_{R}(3)=\{3\}$. Since $1\not \in N_{R}(3)$ and $2 \not \in N_{R}(2)$, we can assume that $p(u_3)=3$ and $p(v_3) \in \{ 2.3\}$. If $p(v_3)=2$ (respectively $3$), then $v_3$ has not neighbors with role~1 (respectively $2$). In the same way, we have that $3\not \in p(V(G))$. Therefore, $p(V(G))=\{2\}$ leads to a contradiction, since role~2 has no loop. \end{proof}

Finally, the last case is the graph obtained by joining two or more copies of $K_{1,m}$, at least one of which is not isomorphic to $K_2$, except for $K_2 \cup K_{1,m}$, whose complementary prisms has a $3$-role assignment, see Lemma~\ref{lemma:k1m,k1}.

\begin{lemma}\label{lemma:bipartite_K_1,t_not_has_3_attribution} Let $G \simeq \bigcup^t_{i=1} K_{1,m_i}$, with $t \geq 2$, $m_1 \geq 2$ and $m_i \geq 1, i=2,\ldots, t$. If $G \not \simeq K_2 \cup K_{1,m}$ with $m \geq 2$, then the complementary prism of $G$ does not have a $3$-role assignment. \end{lemma}

\begin{proof} Before the proof we present son more notation. We consider $G \simeq \bigcup^t_{i=1} G^{(i)}$, with $G^{(i)} \simeq K_{1,m_i}$, with $m_1 \geq 2$, $m_i \geq 1,$ for $i=2,\ldots, t$ and $t \geq 2$. We denote by $V^{(i)}= V(G^{(i)})$, $u_i$ the vertex of degree $m_i$ of $G^{(i)}$ and by $v_i$ one of the leaves of $G^{(i)}$, for $i=1, \ldots, t$. Furthermore, as $m_1 \geq 2$, we denote by $w_1$ a leaf of $G^{(1)}$ other than $v_{1}$. We highlight that in the case of $m_i=1$, we have $V^{(i)}=\{u_i,v_i\}$.

 By contradiction, we assume that $p: V(G\overline{G}) \rightarrow V(R)$ is a $3$-role assignment. If $2\not \in p(V(G))$, then $p(\overline{V(G)})=\{2\}$, a contradiction, since $p(N(\overline{u }_1))=\{p(u_1),2\}$.

First, suppose there is a loop on the role~2. As no leaf of $G$ has role~2, we can assume that $p(u_1)=2$ and consequently $p(\overline{u}_1)=2$. Suppose, without loss of generality, that $p(v_1)=1$ and $p(w_1)=3$. If $1 \in N_R(3)$, then $p(\overline{v}_1)=3$ and $p(\overline{w}_1)=1$. Given this, $R=R_{3,8}$. On the other hand, $p(\overline{V^{(i)}}) \subseteq N_{R}(1)\cap N_{R}(2)\cap N_{R}(3)=\{2 \}$, for $i= 2, \ldots, t$. Therefore, $p(N(\overline{u_1}))=\{2\}$ leads to a contradiction. Otherwise, $1 \not \in N_{R}(3)$. Looking at the neighbor of $\overline{u}_1$, we can assume that $p(\{\overline{u_2}, \overline{v_2}\})=\{1,3\}$. Remembering that $p(v_2) \neq 2$ and $1\not \in N_{R}(3)$, we have that $p(v_2)=p(\overline{v_2})$. Therefore, $p(u_2)=2$, a contradiction, since $2\not \in p(N(u_2))$.

From now on, we can assume that there is no loop on the role~2. Clearly, $G\overline{G}$ is non-bipartite and thus $R \neq R_{3,1}$. Observe that there  a clique of order~4 in $\overline{G}$ (made up of $\{\overline{v_1}, \overline{w_1}, \overline{v_2}, \overline{w_2}\}$ or $\{ \overline{v_1}, \overline{w_1}, \overline{v_2}, \overline{v_3} \}$, where $w_2$ is a leaf of $G^{(2)}$ other than $ v_2$ in $\overline{G})$, thus $R\neq R_{3,7}$. So, $R=R_{3,2}$ or $R_{3,5}$. In both cases $3 \in N_{R}(3)$. We will analyze possible roles for $v_1$.

\textbf{Case~1: $\mathbf{p(v_1)=1}$.}

If $1\not \in N_{R}(1)$, then $p(u_1)=2$ and $p(\overline{v_1})=2$. If $p(w_1)=1$, then $p(\overline{w_1})=2$ and $2 \in N_{R}(2)$, which leads to a contradiction. So, $p(w_1)=3$ and $p(\overline{w_1})=3$. Therefore, for $i=2, \ldots, t,$ we have that $p(\overline{V^{(i)}}) \subseteq N_{R}(2)\cap N_{R}(3)= \{3\}$. As $p(\overline{v_2})=3$, we have that $p(v_2)\in \{2,3\}$. If $p(v_2)=2$, then we have $p(u_2)=1$, which leads to a contradiction, since $p(\overline{u_2})=3$ and $1\not \in N_{R }(3)$. Otherwise, $p(v_2)=3$ and in this case we have $p(u_2)=2$. As $p(\overline{u_2})=3$, there is a leaf $w_2 \in V^{(2)}$, such that $p(w_2)=3$, which leads to a contradiction, since $p(\overline{w}_2)=3$.

If $1\in N_{R}(1)$, then $p(u_1)\in \{1,2\}$, since $ 1\not \in N_{R}(3)$. Observe that in this case, we can exchange, by symmetry, role~3 for role~1 and vice versa. We analyze the possible roles for $p(u_1)\in \{1,2\}$.

If $p(u_1)=1$, then $p(\overline{v_1})=2$. If $p(\overline{u_1})=1$, then for $i=2, \ldots, t$ we have $p(\overline{V^{(i)}})\subseteq N_{R}(1 )\cup N_{R}(2)=\{1\}$. If $p(v_2)=1$, then $p(u_2)=2$ and, since $p(\overline{u_2})=1$, we conclude that there is a leaf $w_2 \in V^{(2) }$, such that $p(w_2)=1$, which leads to a contradiction. Otherwise, $p(v_2)=2$, and we have $p(u_2)=3$, that brings a contradiction. So, $p(\overline{u_1})\neq 1$ and thus $p(\overline{u_1})=2$. On the other hand, $p(\overline{v_1})=2$ and like $v_1w_1\not \in E(G)$, $p(\overline{w_1})\neq 2$. This implies that $p(w_1)\neq 1$, and therefore $p(w_1)=2$. So, $p(\overline{w}_1)=3$, and for $i=2, \ldots, t$, we have that $p(\overline{V^{(i)}}) \subseteq N_{ R}(2)\cap N_{R} (3)=\{3\}$. Again, if $p(v_2)=2$, then $p(u_2)=1$ and we have a contradiction. Otherwise, $p(v_2)=3$, $p(u_2)=2$ and there is a leaf $w_2 \in V^{(2)}$ such that $p(w_2)=1$, which gives the desired contradiction.

If $p(u_1)=2$, then $p(\overline{v_1})=1$. If $p(\overline{u_1})=1$, then we can assume that $p(w_1)=3$ and we have $p(\overline{w_1})=3$, a contradiction, since $ 1 \not \in N_{R}(3)$. Otherwise, $p(\overline{u_1})=3$, then for $i=2, \ldots, t$, we have that $p(\overline{V^{(i)}} )\subseteq N_{R }(1)\cap N_{R}(3)=\{2\}$. So $3\not \in p(N(\overline{u_1}))$, which leads to a contradiction.

\textbf{Case~2: $\mathbf{p(v_1)=2}$.}

If $p(u_1)=1$, then $p(\overline{v_1})=3$. If $p(\overline{u_1})=1$, then for $i=2, \ldots, t$, $p(\overline{V^{(i)}}) \subseteq N_{R}(1 )\cap N_{R}(3)=\{2\}$. Since $t \geq 3$ or $\overline{V^{(2)}}$ has an edge, we get a contradiction, with the fact that $2 \not \in N_{R}(2)$. Otherwise, $p(\overline{u_1})=2$ and for $i=2, \ldots, t$, we have that $p(\overline{V^{(1)}})\subseteq N_{R} (2)\cap N_{R}(3)=\{3\}$. By the same reasoning used before, we obtain $x\in V^{(2)}$, such that $p(x)=1$, a contradiction. In the same way, we get a contradiction if $p(u_1)=3$.

\textbf{Case~3: $\mathbf{p(v_1)=3}.$}


In this case, we can assume that $N_{R}(1)=\{2\}$ and all leaves of $V^{(1)}$ have role~3. As $2 \not \in N_{R}(2)$, $\{p(\overline{v_1}), p(\overline{w_1})\} \neq \{2\}$. So we can assume that $p(\overline{v}_1)=3$. Therefore, $p(u_1)=2$ and $p(\overline{u}_1)=1$. We conclude that $p(\overline{V^{(i)}}) \subseteq N_{R}(1) \cap N_{R}(3)=\{2\}$, for $i=2, \ldots,t$, a contradiction, since $\overline{G-V^{(1)}}$ has at least one edge. \end{proof}

\section{General results}
\label{section_5.3}

The first lemma deals with the union of a complete graph $K_n$ with a complete bipartite graph $K_{1,m}$, for $m \geq 1$. We note that this case includes $K_1 \cup K_{1,m}, K^2_2$ and $K_2\cup K_{1,m}$, excluded in the previous section.

\begin{lemma} \label{lemma:k1m,k1} The complementary prisms of $K_n \cup K_{1,m}$, with $m,n\geq 1$, have a $R_{3,2}$-role assignment. \end{lemma}

\begin{proof} Let $G \simeq K_n \cup K_{1,m}$, with $m,n \geq 1$. Observe that if $n=1$, then $G_0$ consists of only one vertex, otherwise it is empty. We denote by $u_0$ the vertex of degree $m$ of the isomorphic component to $K_{1,m}$. We define $p: V(G \overline{G}) \rightarrow \{1,2,3\}$ as follows.

\begin{minipage}[c]{0.4\linewidth}
For $x\in V(G):$
\vspace{-0.2cm}
$$p(x)=
\begin{cases}
1, \mbox{ if } x \in G_0 ;\\
2, \mbox{ if } x = u_0 ; \\
3, \mbox{ otherwise. }
\end{cases}
\vspace{0.2cm}$$
\end{minipage}
and
\hfill
\begin{minipage}[c]{0.58\linewidth}
\vspace{0.2cm}
$$p(\overline{x})=
\begin{cases}
1, \mbox{ if } x=u_0;\\
2, \mbox{ if } x \not \in N_{G}[u_0];\\
3, \mbox{ if } x \in N_{G}(u_0).
\end{cases}$$
\end{minipage}

We consider $u_1 \in N_{G}(u_0)$. Observe that $p(u_1)=3$ and $p(\overline{u_1})=3.$ We show that $p$ is a $R_{3,2}$-role assignment.

The vertices with role~1 are those belonging to $G_0$ or the vertex $\overline{u_0}$. Let $x \in G_0$, then $p(N(x)) =\{p(\overline{x})\}=\{2\}$. For the vertex $\overline{u_0}$, we have that $p(N(\overline{u}_0))=p(\{u_0\} \cup \overline{V(G)-N_G[u_0]}) =\{2\}$.

The vertices with role~2 are the vertex $u_0$ or those belonging to $\overline{V(G)-N_{G}[u_0]}$. Observe that $V(G)-N_G[u_0]$ is composed of vertices of the clique isomorphic to $K_n$. Clearly, $\{u_0\} \cup \overline{V(G)-N_{G}[u_0]}$ is an independent set. In this case, for each vertex of role~2, we have to point neighbors of role~1 and~3. For the vertex $u_0$, we have $p(\overline{u_0})= 1$ and $p(u_1)=3$. Let $\overline{x}$, with $x \not \in N_{G}[u_0]$, we have $p(\overline{u_0})=1$ and $p(\overline{u_1})=3 $.

The vertices with role~3 are those belonging to $V(G)-(\{u_0\} \cup G_0)$ or to $\overline{N_{G}(u_0)}$. Let $x \in V(G)-(\{u_0\} \cup G_0)$. If $x \in N_{G}(u_0)$, then $p(u_0)=2$ and $p(\overline{x})=3$, this implies that $p(N(x))= \{ 2, 3 \}$. Otherwise, $p(\overline{x})=2$ and $V(G)\neq G_0 \cup N_{G}[u_0],$ then $n\geq 2$, which guarantees that $3 \in p(N(x))$. Let $\overline{x}$, with $x \in N_{G}(u_0)$. We know that there is $y \not \in N_G[u_0]$. By the nature of $G$, $xy \not \in E(G)$. Since $p(x)=3$ and $p(\overline{y})=2$, we have that $p(N(\overline{x}))=\{2,3\}$.\end{proof}


Next, we show some conditions on maximal cliques that assure a $3$-role assignment to the complementary prisms. We present a general result, despite using it only for the case where the clique has order~2.

\begin{lemma} \label{lemma:naobipartido_naosplit_maximum clique} Let $C$ be a maximal clique of the graph $G$. We consider the following conditions:

\begin{enumerate}
\item for every $x \in C$, there is $y \not \in C$, such that $xy \in E(G)$;
\item for every $x \not \in C$, there is $y \not \in C$, such that $xy \not \in E(G)$;
\item for every $x \not \in C$, there is $y \not \in C$, such that $xy \in E(G)$. \end{enumerate}

The complementary prism of $G$ has a $R_{3,4}$-role assignment. \end{lemma}

\begin{proof} We note that Condition~1 implies that $|C|\geq 2$, since $C$ is a maximal clique. We define $p:V(G\overline{G})\rightarrow \{1,2,3\}$, as follows.

\begin{minipage}[c]{0.5\linewidth}
For $x\in V(G):$
\vspace{-0.2cm}
$$p(x)=
\begin{cases}
2, \mbox{ if } x \in C; \\
3, \mbox{ otherwise. }
\end{cases}$$
\end{minipage}
 and
\hfill
\begin{minipage}[c]{0.4\linewidth}
\vspace{0.2cm}
    $$p(\overline{x})=
\begin{cases}
1, \mbox{ if } x \in C;\\
2, \mbox{ otherwise. }
\end{cases}$$
\end{minipage}

\qquad \qquad

We easily see that $p$ is a $R_{3,4}$-role assignment.\end{proof}




The next lemma presents conditions on a leaf and its neighborhood. Although we use this lemma only in the case of non-bipartite graphs, we prefer to present it now for consistency with the next two lemmas. In fact, the assignment defined in the three proofs is the same, but we split it in different cases for the sake of simplicity. Remember that $G_1$ is the set of leaves of $G$, that is, of the vertices of degree $1$.

\begin{lemma} \label{lemma:nobipartite_semi-isolated_clique_at least1verticalleaf} Let $f$ be a leaf of the graph $G$ with $N_{G}(f)=\{d\}$ such that $N_{G}(d)$ is an independent set. If there is $a \in N_{G}(d)$, such that $a\not \in G_1$ and $V(G)\neq N_{G}[a]\cup \{f\}$, then the complementary prism of $G$ has a $R_{3,4}$-role assignment. \end{lemma}

\begin{proof} Remember that $f\in G_1 \cap N_{G}(d)$. We consider $I=G_1 \cap N_{G}(d)$. We define $p:V(G\overline{G})\rightarrow \{1,2,3\}$, as follows.

\begin{minipage}[c]{0.5\linewidth}
For $x\in V(G):$
\vspace{-0.2cm}
$$p(x)=
\begin{cases}
1, \mbox{ if } x\in G_0 \mbox{ or } x\in I;\\
2, \mbox{ if } x=a \mbox{ or } x=d; \\
3, \mbox{ otherwise. }
\end{cases}
\vspace{0.2cm}$$
\end{minipage}
and
\hfill
\begin{minipage}[c]{0.5\linewidth}
\vspace{0.2cm}
$$p(\overline{x})=
\begin{cases}
1, \mbox{ if } x=a; \\
2, \mbox{ if } x\not \in N_{G}[a];\\
3, \mbox{ if } x\in N_{G}(a).
\end{cases}$$
\end{minipage}

By hypothesis, $a\not \in G_1$, so there is $u\in N_{G}(a)$, $u\neq d$. As $N_{G}(d)$ is an independent set, we have that $ud\not \in E(G)$. We consider $w\in V(G)-(N_{G}[a] \cup \{f\}).$ We have that $p(u)=p(\overline{u})=3$, $p (w) \in \{1,3\}$ and $p(\overline{w})=2$. We show that $p$ is a $R_{3,4}$-role assignment.

The vertices with role~1 are those belonging to $G_0 \cup I$ or the vertex $\overline{a}$. Let $x\in G_0$, we have that $p(N(x))=\{p(\overline{x})\}=\{2\}$. Let $x \in I$, we have that $p(N(x)) = p(\{\overline{x},d\})=\{2\}$. For the vertex $\overline{a}$, we have that $p(N(\overline{a}))=p(\{a\}\cup \overline{V(G)-N_{G}[a] })=\{2\}$.

The vertices with role~2 are $a$, $d$ and those belonging to $\overline{V(G)-N_{G}[a]}$. For the vertex $a$, we have $p(\overline{a})=1$, $p(d)=2$ and $p(u)=3$. For the vertex $d$, we have $p(\overline{d})=3$, $p(a)=2$ and $p(f)=1$. Let $\overline{x}$, with $x\in G_0 \cup I$, so $p(x)=1$ and $p(\overline{u})=3$. For the vertex $\overline{f}$, we have $p(\overline{w})=2$. If $x \neq f$, we have $p(\overline{f})=2$. Let $\overline{x}$, with $x\not \in N_{G}[a]$ and $x\not \in G_0 \cup I$, we have $p(x)=3$, $p( \overline{a})=1$ and $p(\overline{f})=2$.

The vertices with role~3 are those belonging to $V(G)-(\{a,d\} \cup G_0 \cup I)$ or to $\overline{N_{G}(a)}$. Let $x \in V(G)$ be such that $x \neq a,d$ and $x \not \in G_0 \cup I$. If $x\in N_{G}(a)$, then we have $p(\overline{x})=3$ and $p(a)=2.$ Otherwise, $x\not \in N_{G}[a]$ and we have $p(\overline{x})=2$. As $x\not \in G_0 \cup I$, there is $y\in N_{G}(x)$ such that $y\neq a,d$, that is $p(y)=3$. Let $\overline{x}$, with $x \in N_{G}(a)$. For the vertex $\overline{d}$, we have $p(d)=2$ and $p(\overline{u})=3$. If $x\neq d$, we have $p(\overline{f})=2$ and $p(\overline{d})=3$. \end{proof}

In the next lemma, the graph does not necessarily have a leaf, however it has a vertex that behaves like a leaf.

\begin{lemma} \label{lemma:nobipartite_withisolated_withclick} Let $a$ and $b$ be two non-isolated vertices of $G$, such that $ab\not \in E(G)$, $N_{G}(a )$ is an independent set, $N_{G}(b) \subseteq N_{G}(a)$, $N_{G}(b) \neq N_{G}(a)$ and $V(G) \neq N_{G}[a]\cup \{b\}$. The complementary prism of $G$ has a $R_{3,4}$-role assignment. \end{lemma}

\begin{proof} It follows from the hypothesis that $|N_{G}(a)|\geq 2$. We consider $b$ a vertex, with the smallest degree, that satisfies the lemma hypotheses for some $a \in V(G)$ and $I=\{x\in V(G) \mid N_{G}(x)=N_ {G}(b)\}$. Observe that $b \in I$, $I$ is an independent set and $I\cap N_{G}(b)= \emptyset$. Clearly, $a \not \in I$ and for all $x \in I$, $ax \not \in E(G)$. We define $p:V(G\overline{G})\rightarrow \{1,2,3\}$, as follows.

\begin{minipage}[c]{0.5\linewidth}
For $x\in V(G):$
\vspace{-0.2cm}
$$p(x)=
\begin{cases}
1, \mbox{ if } x \in G_0 \mbox{ or } x\in I;\\
2, \mbox{ if } x = a \mbox{ or } x\in N_{G}(b); \\
3, \mbox{ otherwise. }
\end{cases}
\vspace{0.2cm}$$
\end{minipage}
and
\hfill
\begin{minipage}[c]{0.5\linewidth}
\vspace{0.2cm}
$$p(\overline{x})=
\begin{cases}
1, \mbox{ if } x = a;\\
2, \mbox{ if } x\not \in N_{G}[a];\\
3, \mbox{ if } x\in N_{G}(a).
\end{cases}$$
\end{minipage}

We consider $u \in N_{G}(a)-N_{G}(b)$ and $v\in N_{G}(b) \subseteq N_{G}(a)$ and $w \in V( G)-(N_{G}[a] \cup \{b\})$. We have that $p(u)=3$, $p(\overline{u})=3$, $p(v)=2$, $p(\overline{v})=3$, $p(w )\in \{1,3\}$ and $p(\overline{w})=2$. We show that $p$ is a $R_{3,4}$-role assignment.

The vertices with role~1 are those belonging to $G_0\cup I$ or the vertex $\overline{a}$. Let $x \in G_0$, we have that $p(N(x))=\{p(\overline{x})\}=\{2\}$. Let $x \in I$, we have that $p(N(x))=\{p(\overline{x})\}\cup p(N_{G}(b))=\{2\}$. For the vertex $\overline{a}$, we have that $ p(N(\overline{a}))=\{p(a)\} \cup p(\overline{V(G)-N_{G} [a]})=\{2\}$.

The vertices with role~2 are $a$ or those belonging to $N_{G}(b)$ or to
 $\overline{V(G)-N_{G}[a]}$. For the vertex $a$, we have $p(\overline{a})=1$, $p(v)=2$ and $p(u)=3$. Let $x \in N_{G}(b)$, we have $p(\overline{x})=3$, $p(a)=2$ and $p(b)=1.$ Let $\overline {x}$, with $x\not \in N_{G}[a]$, we have $p(\overline{a})=1$. For the vertex $\overline{b}$, we have $p(\overline{w})=2$ and $p(\overline{u})=3$. Otherwise, $\overline{x}\neq \overline{b}$, so $p(\overline{b})=2$. If $x \in G_0 \cup I$, then $p(\overline{u})=3$. Otherwise, we have $p(x)=3$.

The vertices with role~3 are those belonging to $V(G)-(G_0 \cup I \cup \{a\} \cup N_{G}(b))$ or to $\overline{N_{G}( a)}$. Let $x\in V(G)$ be such that $x \neq a$ and $x \not \in G_0 \cup I \cup N_{G}(b)$. If $x\in N_{G}(a)$, then $p(a)=2$ and $p(\overline{x})=3$. Otherwise, we have $p(\overline{x})=2$. Observe that if $N_{G}(x)\subseteq N_{G}(b)$, then the vertices of $x$ and $b$ satisfy the lemma conditions, this leads to a contradiction with the minimality of $b $. In fact, by assumptions, $x,b \not \in G_0$, $bx \not \in E(G)$, $N_{G}(b)$ is an independent set, $N_{G}(x )\neq N_{G}(b)$ since $x\not \in I$ and as $a\not \in N_{G}[b] \cup \{x\}$, we have $V(G )\neq N_{G}[b] \cup \{x\}$. So $N_{G}(x) \not \subseteq N_{G}(b)$ and there is $y \in N_{G}(x)-N_{G}(b)$. Knowing that $y\neq a$, $y\not \in G_0$ and $y\not \in I$, then $p(y)=3$ and $3\in p(N(x))$. Let $\overline{x}$, with $x\in N_{G}(a)$. If $x \in N_{G}(b)$, then $p(x)=2$. If $x \not \in N_{G}[b]$, then $p(\overline{b})=2$. Since $N_{G}(a)$ is an independent set, $|N_{G}(a)|\geq 2$ and $p(\overline{N_{G}(a)})=3$, we have than $3 \in p(N(\overline{x}))$. \end{proof}

\begin{figure} [ht!]
	\centering
	\includegraphics[scale=1.05]{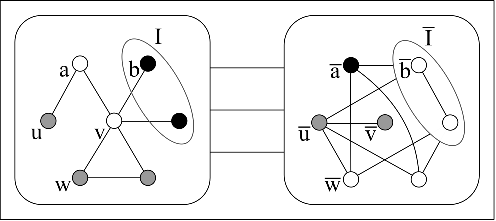} 
	\caption{Example of the complementary prism of a graph with $R_{3,4}$-role assignment.}
	\label{lema:semcliquetamanho3_conjindependente}
\end{figure}

To obtain the last lemma of the section, we need the following lemma.

\begin{lemma} \label{lemma_clique3_greaterequal3} Let $G$ be a graph of order $n\geq 6$, with no clique greater than or equal to $3$. If there are $a,b\in V(G)$, $ab\not \in E(G)$, $N_{G}(a) \cap N_{G}(b)\neq \emptyset$, $ N_{G}(a)-N_{G}(b) \neq \emptyset$, $N_{G}(b)-N_{G}(a) \neq \emptyset$ and $V(G)=N_ {G}[a] \cup N_{G}[b]$, then the complementary prisms of $G$ has a $R_{3,4}$-role assignment. \end{lemma}

\begin{proof} We note that $a$ and $b$ are not leaves in $G$. Since $G$ has no clique of order~3, we have that $N_{G}(a)$ is an independent set. By hypothesis, there is $u \in N_{G}(a) \cap N_{G}(b)$. If there is a leaf $f$ of $G$, then by symmetry we can assume that $f\in N_{G}(a)$. Clearly, $u \not \in G_1$ and $V(G) \neq N_{G}[u] \cup \{f\}$ since any vertex in $N_{G}(b)-N_{G }(a)$ is not a neighbor of $u$ and $N_{G}(b)-N_{G}(a)\neq \emptyset$. Therefore, the result follows from the Lemma~\ref{lemma:nobipartite_semi-isolated_clique_at least1verticalleaf}. Suppose $G$ has no leaf. Therefore, we can assume that $|N_{G}(b)| \geq 3$. In fact, otherwise $|N_{G}(b)|=2$ and $|N_{G}(a)|=2$, implies that $n=5$, a contradiction. We show that $C=\{a,u\}$ satisfies the conditions of Lemma~\ref{lemma:naobipartido_naosplit_maximum clique}. By hypothesis, $C$ is a maximal clique. Conditions~1 and~3 follow from the fact that $G$ has no leaves and has no clique order~3. Condition~2 remains to be checked. We consider $v \in N_{G}(a)-N_{G}(b)$. Clearly, $v\not \in C$. For vertex $b$, $bv\not \in E(G)$. If $x\in N_{G}(a)\cap N_{G}(b)$, then $xv\not \in E(G)$. If $x\in N_{G}(a)-N_{G}(b)$, then $bx\not \in E(G)$. If $x\in N_{G}(b)-N_{G}(a)$, then as $|N_{G}(b)|\geq 3$, and $N_{G}(b)$ is an independent set, there is $y\in N_{G}(b)$, $y\neq x,u$ and $xy \not \in E(G)$. Therefore, $\{a,u\}$ satisfies the conditions of Lemma~\ref{lemma:nobipartite_semi-isolated_clique_at least1verticalleaf} and the result follows. \end{proof}

The last lemma of the section deals with a generalization of the Lemma~\ref{lemma:nobipartite_withisolated_withclick} but in the context of the Lemma~\ref{lemma_clique3_greaterequal3}. Observe that in this case, the neighborhood of any vertex is an independent set.

\begin{lemma} \label{lemma:nobipartite_without clique_comisolated} Let $G$ be a graph of order $n \geq 6$, with no clique of order greater than or equal to $3$. If there are $a,b \in V(G)$, such that $ab\not \in E(G)$, $N_{G}(a) \cap N_{G}(b) \neq \emptyset$, $N_{G}(a)- N_{G}(b)\neq \emptyset$ and $V(G)\neq N_{G}[a] \cup \{b\}$, then the complementary prisms of $G$ has a $R_{3,4}$-role assignment. \end{lemma}

\begin{proof}
Observe that if $N_{G}(b)-N_{G}(a) = \emptyset$, then the result follows from Lemma~\ref{lemma:nobipartite_withisolated_withclick}. Therefore, we can assume that $N_{G}(b)-N_{G}(a) \neq \emptyset$ and that no pair of vertices satisfies the conditions of Lemma~\ref{lemma:nobipartite_withisolated_withclick}. By Lemma~\ref{lemma_clique3_greaterequal3}, we can assume that $V(G) \neq N_G[a] \cup N_G[b]$.

We denote by $I=\{x\in V(G), \mid N_{G}(x)=N_{G}(b)\}$ and define $p:V(G \overline{G})\rightarrow \{1,2,3\}$ as follows.

For $x\in V(G)$:

\begin{minipage}[c]{0.5\linewidth}
$$p(x)=
\begin{cases}
1, \mbox{ if } x \in G_0 \mbox{ or } x \in I; \\
2, \mbox{ if } x =a \mbox{ or } x \in N_{G}(b); \\
3, \mbox{ otherwise. }
\end{cases}$$
\vspace{0.05cm}
\end{minipage}
and
\hfill
\begin{minipage}[c]{0.5\linewidth}
$$p(\overline{x})=
\begin{cases}
1, \mbox{ if } x=a; \\
2, \mbox{ if } x \not \in N_{G}[a];\\
3, \mbox{ if } x \in N_{G}(a).
\end{cases}$$
\vspace{0.05cm}
\end{minipage}

We note that $I$ is an independent set, $a \not \in I$ and $b \in I$. We consider $u\in N_{G}(a)\cap N_{G}(b)$, $v\in N_{G}(a)-N_{G}(b)$ and $w\in V( G)-(N_{G}[a] \cup N_{G}[b])$.
Observe that $p(u)=2$, $p(\overline{u})=3$, $p(v)=3$, $p(\overline{v})=3$, $p(w )\in \{1,3\}$ and $p(\overline{w})=2$. We show that $p$ is a $R_{3,4}$-role assignment.

The vertices with role~1 are those belonging to $G_0 \cup I$ or the vertex $a$. Let $x \in G_0$. Clearly, $p(N(x))=\{p(\overline{x})\}=\{2\}$. Let $x \in I$, we have that $p(N(x))=p(\{\overline{x}\}\cup N_{G}(b))=\{2\}$, since $x \neq a$ and $xa\not \in E(G)$. For the vertex $\overline{a}$, we have $ N(\overline{a})=\{a\} \cup \overline{V(G)-N_{G}[a]}$, so $p (N(\overline{a}))=\{2\}$.

The vertices with role~2 are the vertex $a$ or those belonging to $N_{G}(b)\cup \overline{V(G)-N_{G}[a]}$. For the vertex $a$, we have $p(\overline{a})=1$, $p(u)=2$ and $p(v)=3$. Let $x \in N_{G}(b)$, we have $p(b)=1$. If $x\in N_{G}(a)$, then we have $p(a)=2$ and $p(\overline{x})=3$. Otherwise, $x \not \in N_{G}[a]$ and $p(\overline{x})=2$. If $N_G(x) \subseteq N_G(u)$, then we show a contradiction with the fact that no pair of vertices satisfies Lemma~\ref{lemma:nobipartite_withisolated_withclick}. In fact, as $x \in N_G(b)-N_G(a)$, we have that $x \neq u, v$. By hypothesis, the neighborhood of any vertex forms an independent set. So, $uv, ux\not \in E(G)$, $N_G(x)$ is an independent set, $N_{G}(x)\subseteq N_{G}(u)$, $a \in N_{G}(u)-N_{G}(x)$ and $v \not \in N_{G}[u]\cup \{x\}$. Therefore, $u$ and $x$ satisfy the conditions of Lemma~\ref{lemma:nobipartite_withisolated_withclick}. Otherwise, $N_{G}(x)\nsubseteq N_{G}(u)$, that is, there is $y \in N_{G}(x)-N_{G}(u)$. Remember that $p(\overline{x})=2$. We show that $p(y)=3$. As $xy\in E(G)$, we have that $y \not \in G_0$, $y \neq a$ and $y\not \in N_G(b)$, since $G$ has no clique of order~3. As $uy\not \in E(G)$ and $bu\in E(G)$, we have $y\not \in I$. Therefore, $p(y)=3$ and $3 \in p(N(x))$. Let $\overline{x}$, with $x \not \in N_{G}[a]$, we have $p(\overline{a})=1$. For the vertex $b$, we have $p(\overline{w})=2$ and $p(\overline{v})=3$. If $x\in G_0 \cup I$, $x \neq b$, we have $p(x)=1, p(\overline{b})=2$ and $p(\overline{v})=3$. If $x\in N_{G}(b)$, as $x \not \in N_{G}[a]$, then $x\neq u$ and $xu\not \in E(G)$. Thus, $p(x)=2$ and $p(\overline{u})=3$, we can conclude that $p(N(\overline{x}))=\{1,2,3\} $. If $x \not \in N_{G}(b)$ and $x \not \in G_0 \cup I$, we have $p(x)=3$ and $p(\overline{b})=2$ , which gives us the desired conclusion.

The vertices with role~3 are those belonging to $V(G)-(\{a\} \cup G_0 \cup I \cup N_{G}(b))$ or to $\overline{N_{G}( a)}$. Let $x \in V(G)$, $x \neq a$, $x \not \in G_0 \cup I \cup N_{G}(b)$. If $x \in N_{G}(a)$, then $p(a)=2$ and $p(\overline{x})=3$. Otherwise, $ x \not \in N_{G}(a)$ and $p(\overline{x})=2$. We show that if $N_G(x) \subseteq N_G(b)$, then the vertices of $b$ and $x$ satisfy the Lemma~\ref{lemma:nobipartite_withisolated_withclick}, a contradiction. In fact, $bx \not \in E(G)$, $N_G(b)$ is an independent set, $N_{G}(b) \neq N_{G}(x)$ and $a \not \in N_{G}[b]\cup \{x\}$. Therefore, $N_{G}(x)\not \subseteq N_{G}(b)$, that is, there is $y\in N_{G}(x)-N_{G}(b)$. Clearly, $y \not \in G_0$, $y \not \in I$, $y \neq a$ and $y \not \in N_{G}(b)$, so $p(y)=3 $. Let $\overline{x}$, with $x \in N_{G}(a)$. We note that, $x \not \in \{a\} \cup G_0 \cup I$. If $x \in N_{G}(b)$, then $p(x)=2$ and $p(\overline{v})=3$ provide the necessary roles for the neighborhood of $x$. If $x \not \in N_{G}(b)$, then $p(\overline{b})=2$ and $p(x)=3$. We conclude that $p(N(\overline{x}))=\{2,3\}$. \end{proof}

\section{Results for complementary prisms of bipartite graphs}
\label{section_5.6}

In this section, we highlight the complementary prisms of bipartite graphs with a $3$-role assignment. The first case is the connected bipartite graph whose smallest set of the bipartition has two vertices and we obtain the graph $R_{3,2}$ that is the path graph of size~3 with a loop at role~3.

\begin{lemma} \label{lemma:bipartite_graph_with_universal_vertice} Let $G$ be a connected bipartite graph, such that the smallest partition has two vertices. The complementary prisms of $G$ has a $R_{3,2}$-role assignment.\end{lemma}

\begin{proof} We consider a bipartite partition $(A,B)$, with $|A|\geq |B|=2$. Let $B=\{a,b\}$. Since $G$ is a connected bipartite graph, there is $u \in N_{G}(a) \cap N_{G}(b)$ and $V(G)=N_{G}[a]\cup N_ {G}[b]$.
We define $p:V(G\overline{G})\rightarrow \{1,2,3\}$, as follows.

\begin{minipage}[c]{0.4\linewidth}
For $x\in V(G):$
\vspace{-0.2cm}
$$p(x)=
\begin{cases}
1, \mbox{ if } x = u;\\
2, \mbox{ if } x \in B; \\
3, \mbox{ otherwise. }
\end{cases}$$
\vspace{0.1cm}
\end{minipage}
and
\hfill
\begin{minipage}[c]{0.4\linewidth}
$$p(\overline{x})=
\begin{cases}
2, \mbox{ if } x=u;\\
3, \mbox{ otherwise. }
\end{cases}$$
\end{minipage}

One can easily see that $p$ is a $R_{3,2}$-role assignment. \end{proof}




We will see in the following theorem that the complementary prism of a bipartite graph with a non-trivial connected component that is not isomorphic to a complete bipartite graph has a $3$-role assignment.

\begin{theorem}\label{finaltheorem:bipartite_not_connected_seuniversal_howoremisolated1} Let $G$ be a bipartite graph. If there is a non-trivial connected component of $G$ that is not isomorphic to $K_{n,m}$, with $n,m \geq 1$, then the complementary prism of $G$ has a $3$-role assignment. \end{theorem}

\begin{proof} Let $G'$ be a non-trivial connected component of $G$, not isomorphic to $K_{n,m}$ with $n,m\geq 1$. We consider $(A,B)$ a bipartite partition of $G'$. There exist $a,b\in V(G')$, $a \in A$ and $b\in B$, such that $ab\not \in E(G)$. Since $G'$ is connected, there is a shortest path $u_0u_1 \ldots u_s$, with $a=u_0$ and $b=u_s$, such that $u_i\in A$, for every even $i$, and $u_i \in B$ for every odd $i$. We note that $s$ is odd and $s \geq 3$. We obtain that $au_3\not \in E(G)$, with $u_1 \in N_{G}(a)\cap N_{G}(u_2)$ and $u_3 \in N_{G}(u_2)-N_ {G}(a)$. By the Lemma~\ref{lemma:nobipartite_without clique_comisolated}, we can suppose $V(G)=N_{G}[u_2] \cup \{a\}$ or $s\leq 5$. If $G$ is connected, then in both cases $G$ satisfies the Lemma~\ref{lemma:bipartite_graph_with_universal_vertice}. If $G$ is not connected, then $s=5$ and $G \simeq K_1 \cup P_4$ satisfies the Lemma~\ref{lemma:nobipartite_semi-isolated_clique_at least1verticalleaf}.\end{proof}

With these two previous results, it is possible to characterize the complementary prisms with a $3$-role assignment for bipartite graphs with isolated vertices, as we will see in the following theorem.

\begin{theorem} \label{theorem:bipartite_nonconnected_withisolated} Let $G$ be a bipartite graph with isolated vertices. The complementary prism of $G$ has a $3$-role assignment if and only if there exists a non-trivial connected component, not isomorphic to $K_{n,m}$, with $n,m\geq 1$ or $G\simeq K_1 \cup K_{1,m}$, $m\geq 1$. \end{theorem}

\begin{proof} $( \Rightarrow )$ By contradiction, we assume that every non-trivial connected component is isomorphic to $K_{n,m}$, with $n,m\geq 1$. It follows from the Lemma~\ref{lemma:k_n,mUk_1_nao_has_3attribution} that $G\simeq K_1 \cup K_{1,m}$ for some $m \geq 1$.

$(\Leftarrow)$ It follows from Lemma~\ref{lemma:k1m,k1} that the complementary prisms of $K_1 \cup K_{1,m}$, with $m\geq 1$, has a $3$-role assignment. So we can assume that $G \not \simeq K_1 \cup K_{1,m}$. Thus, there is a non-trivial connected component of $G$ that is not isomorphic to $K_{n,m}$ with $n,m\geq 1$ and the result follows from Theorem \ref{finaltheorem:bipartite_not_connected_seuniversal_howoremisolated1}.\end{proof}

We need a few more results to get the characterization of complementary prisms of bipartite graphs without isolated vertices with $3$-role assignment.
By Theorem~\ref{finaltheorem:bipartite_not_connected_seuniversal_howoremisolated1}, we can assume that $G$ is the union of complete bipartite graphs. The next lemma solves the case when the graph has a connected component isomorphic to $K_{n,m}$, with $n,m \geq 2$.

\begin{lemma} \label{lemma:bipartite_nonconnected_semiisolated_adjacent} Let $G$ be a bipartite graph with no isolated vertices. If there is $G'$ a connected component of $G$ that is isomorphic to $K_{n,m}$ with some $n,m \geq 2$, then the complementary prism of $G$ has a $3$-role assignment. \end{lemma}

\begin{proof} Let $G'$ be a connected component of $G$ isomorphic to $K_{n,m}$ for some $n,m \geq 2$. If $G$ is connected and $n=2$ or $m=2$, then it follows from Lemma~\ref{lemma:bipartite_graph_with_universal_vertice}. Let $uv \in E(G')$. We show that $\{u,v\}$ satisfies the Lemma~\ref{lemma:naobipartido_naosplit_maximum clique}. Since $G$ is bipartite, $\{u,v\}$ is a maximal clique. Since $n,m \geq 2$, Condition~1 is satisfied. Let $x \neq u,v$. Since $x$ is not an isolated vertex, there exists $y\in V(G)$, such that $xy\in E(G)$. Since $G'$ has no leaves, we can assume that $y\neq u,v$. Thus, Condition~3 is satisfied, and since $G$ is not connected or $n,m \geq 3$, Condition~2 is satisfied. \end{proof}

From this lemma, it remains the case where $G$ is the union of copies of $K_{1,m}$, including graph $K_2$. Recall that the Figure~\ref{grafo_K_2^3} and the Lemmas~\ref{lemma:prisma_Kn_nao_tem}, ~\ref{lemma:split_(K_2)^t}, ~\ref{lemma:k1m,k1} resolve all the cases when the graph $G$ is the union of copies of $ K_2$. 
Therefore, we can assume that $G$ has at least one connected component isomorphic to $K_{1,m}$, with $m\geq 2$. In view of Lemmas~\ref{lemma:bipartite_K_1,t_not_has_3_attribution} and~\ref{lemma:k1m,k1}, the only remaining case is the complementary prisms of $K_{1,m}$ that we will see in the next lemma.

\begin{lemma} \label{lemma:k1m} The complementary prisms of $K_{1,m}$, with $m \geq 2$, has a $R_{3,2}$-role assignment. \end{lemma}

\begin{proof} We denote by $u_0$ the single vertex of $K_{1,m}$ of degree $m$ and $N_{G}(u_0)=\{u_1, \ldots, u_m\}$. We define $p: V(G \overline{G}) \rightarrow \{1,2,3\}$, as follows.

For $x\in V(G):$

\begin{minipage}[c]{0.5\linewidth}
$$p(x)=
\begin{cases}
1, \mbox{ if } x= u_1; \\
2, \mbox{ if } x = u_0; \\
3, \mbox{ if } x \in \{u_2, \ldots, u_m\}.
\end{cases}$$
\end{minipage}
and
\hfill
\begin{minipage}[c]{0.5\linewidth}
$$p(\overline{x})=
\begin{cases}
1, \mbox{ if } x=u_0;\\
2, \mbox{ if } x=u_1;\\
3, \mbox{ if } x \in \{u_2, \ldots, u_m \}.
\end{cases}$$
\end{minipage}

\ \\
One can easily see that $p$ is a $R_{3,2}$-role assignment. \end{proof}




The characterization ends with Theorem~\ref{general_bipartite_theorem}, by covering the complementary prisms of bipartite graphs.

\begin{theorem}\label{general_bipartite_theorem} Let $G$ be a bipartite graph. The complementary prism of $G$ has a $3$-role assignment if and only if $G$ is not isomorphic to one of the following graphs:

\begin{enumerate}
\item $K_2$, nor $\overline{K_n}$, with $n\geq 2$;
\item $G$ has isolated vertices and every non-trivial connected component of $G$ is isomorphic to $K_{n,m}$, for some $n,m\geq 1$, with the exception of $K_1 \cup K_{ 1,m}$, with $m\geq 1$;
\item $K_2^t$, with $t\geq 4$;
\item $\bigcup^t_{i=1} K_{1,m_i}$, with $t\geq 3$, $m_1 \geq 2$, $m_i \geq 1$, $i=2, \ldots, t$;
\item $K_{1,m_1}\cup K_{1,m_2}$, with $m_1,m_2 \geq 2$. \end{enumerate} \end{theorem}

\begin{proof} ($\Rightarrow$) Follows from Lemmas~\ref{lemma:prisma_Kn_nao_tem},~\ref{lemma:k_n,mUk_1_nao_has_3attribution},~\ref{lemma:split_(K_2)^t} and~\ref{lemma:bipartite_K_1,t_not_has_3_attribution}.

($\Leftarrow$) If $G$ has isolated vertices, then the result follows from Theorem~\ref{theorem:bipartite_nonconnected_withisolated}. So suppose $G$ has no isolated vertices. We can assume, by contradiction, that the complementary prisms of $G$ do not admit a $3$-role assignment. By Theorem~\ref{finaltheorem:bipartite_not_connected_seuniversal_howoremisolated1}, every connected component of $G$ is isomorphic to $K_{n,m}$, for some $n,m \geq 1$. It follows from the Lemma~\ref{lemma:bipartite_nonconnected_semiisolated_adjacent} that every connected component is isomorphic to $K_{1,m}$, for some $m\geq 1$. If $G\simeq K_2^t$, for some $t\geq 1$, then, by hypothesis (items~1 and~3), $t=2$ or $t=3$. The contradiction follows from the Lemmas~\ref{lemma:k1m,k1} and Figure~\ref{grafo_K_2^3}. So, $G\simeq \bigcup^t_{i=1} K_{1,m_i}$, with $t\geq 1$, $m_1 \geq 2$, $m_i \geq 1$, $i=2, \ldots, t$. By hypothesis (items $4$ and $5$), $G \simeq K_2\cup K_{1,m}$ or $G\simeq K_{1,m},$ for some $m \geq 2$. The contradiction follows from the Lemmas~\ref{lemma:k1m,k1} and~\ref{lemma:k1m}. \end{proof}

\section{Results for complementary prisms of non-bipartite graphs}
\label{section_5.8}

In this section, we show that the complementary prism of any graph that is neither bipartite nor its complementary graph always has a $3$-role assignment. 
The following theorem considers the case of a non-bipartite graph with no clique of order~3.  we obtain the graph $R_{3,2}$ that is the path graph of size~3 with a loop at role~3. 

Remember that the graph $R_{3,2}$ is the path graph of size~3 with a loop at role~2 and role~3.

\begin{theorem} \label{theorem:nobipartite_withverticeisolated_withoutclick} Let $G$ be a non-bipartite graph. If $G$ does not have a clique of order~3, then the complementary prism of $G$ has a $3$-role assignment. \end{theorem}

\begin{proof} As $G$ is non-bipartite and does not have a clique of order~3, then there is a cycle $u_0 u_1 \ldots u_s$, with  $u_0 =u_s$ and $s\geq 5$. Suppose that $s$ is the smallest possible. We consider $a=u_0$ and $b=u_2$, so $ab\not \in E(G)$, $u_1\in N_{G}(a)\cap N_{G}(b)$, $u_ {s-1}\in N_{G}(a)-N_{G}(b)$ and $u_3 \not \in N_{G}[a] \cup \{b\}$. From the Lemma~\ref{lemma:nobipartite_without clique_comisolated}, we can assume that $s=5$. We emphasize that $C_5$ is the only graph of order~5 non-bipartite and does not have a clique of order~3. Therefore, $G$ is isomorphic to  $C_5$ and the result follows from Figure~\ref{figure_C5_C5}.\end{proof}

From the previous theorem, we can assume that $G$ has a clique of order~3. The next lemma proposes a condition for such a clique. This lemma will be used frequently throughout the following results.

\begin{lemma}\label{lemma:connected_semiisolated_notbipartite_withtriangle_with1ornoverticalleaf} Let $\{a,b,c\}$ be a clique of a graph $G$ such that $a$ and $b$ have no leaf neighbors. If, for every $x \neq a,b$, there is $y \in V(G)-(N_{G}[a] \cap N_{G}[b])$, $y \neq x$ such that, $xy \not \in E(G)$, then the complementary prisms of $G$ has a $R_{3,4}$-role assignment. \end{lemma}

\begin{proof} We define $p:V(G\overline{G})\rightarrow \{1,2,3\}$, as follows. For $x\in V(G):$

\begin{minipage}[c]{0.4\linewidth}
\vspace{-0.2cm}
$$p(x)=
\begin{cases}
1, \mbox{ if } x \in G_0; \\
2, \mbox{ if } x = a \mbox{ or } x=b; \\
3, \mbox{ otherwise. }
\end{cases}$$
\end{minipage}
and
\hfill
\begin{minipage}[c]{0.6\linewidth}
$$p(\overline{x})=
\begin{cases}
1, \mbox{ if } x = a \mbox{ or } x=b;\\
2, \mbox{ if } x \not \in N_{G}[a]\cap N_{G}[b];\\
3, \mbox{ if } x\in N_{G}(a) \cap N_{G}(b).
\end{cases}$$
\end{minipage}

\vspace{0.2cm}
We show that $p$ is a $R_{3,4}$-role assignment.

The vertices with role~1 are those belonging to $G_0$ or the vertices $\overline{a}$ and $\overline{b}$. Let $x\in G_0$, we have that $p(N(x))=\{p(\overline{x})\}=\{2\}$. Let $\overline{x}$, with $x \in \{a,b\} $, we have that $N(\overline{x}) \subseteq \{x\} \cup (V(G)-( N_{G}[a]\cap N_{G}[b]))$, so $p(N(\overline{x}))=\{2\}.$

The vertices with role~2 are $a$, $b$ or those belonging to $\overline{V(G)-(N_{G}[a] \cap N_{G}[b]))}$. For the vertex $a$, we have $p(\overline{a})=1$, $p(b)=2$ and $p(c)=3$. We proceed in the same way for the vertex $b$. Let $\overline{x}$, with $x\not \in N_{G}[a] \cap N_{G}[b]$. By hypothesis, there is $y\in V(G)-(N_{G}[a]\cap N_{G}[b])$, such that $xy\not \in E(G)$. So, $2\in p(N(\overline{x}))$, since $p(\overline{y})=2$. Since $\overline{a}$ or $\overline{b}$ belongs to $N(\overline{x})$, we have $1\in p(N(\overline{x}))$. For $\overline{x}$, with $x \in G_0$, we get $p(\overline{c})=3$ and $p(N(\overline{x}))=\{1,2,3\}$. Otherwise, $p(x)=3$ and $p(N(\overline{x}))=\{1,2,3\}$.

The vertices with role~3 are those belonging to $V(G)-(\{a,b\} \cup G_0)$ or to $\overline{N_{G}(a) \cap N_{G}(b)}$. Let $x\in V(G)-(\{a,b\} \cup G_0)$. If $x\in N_{G}(a)\cap N_{G}(b)$, then $p(a)=2$ and $p(\overline{x})=3$. Otherwise, $p(\overline{x})=2$. Since $x$ is not an isolated vertex, there is $y\in V(G)$, such that $xy\in E(G)$. By the assumptions, the vertices of $a$ and $b$ have no leaf neighbors, so we can choose $y\neq a,b$. Therefore, $p(y)=3$ and $p(N(x))=\{2,3\}$. Let $\overline{x}$, with $x\in N_{G}(a)\cap N_{G}(b)$. By hypothesis, we have that $2 \in p(N(\overline{x}))$. As $p(x)=3$, we have that $p(N(\overline{x}))=\{2,3\}$. \end{proof}

 The following lemma deals with the case when a clique has at least two leaf neighbors.

\begin{lemma} \label{lemma:connection_semiisolated_naobipartite_withtriangle_2leaves} If the graph $G$ has a clique of order~3, with at least two leaf neighbors, then the complementary prisms of $G$ has a $R_{3,4} $-role assignment. \end{lemma}

\begin{proof} For $x\in V(G)$, we denote by $F_{x}=N_{G}(x) \cap G_1$ the set of leaves neighboring the vertex $x$. Let $\{a,b,c\}$ be a clique of $G$, such that $|F_{a} \cup F_{b}| \geq 2$. We define $p:V(G\overline{G})\rightarrow \{1,2,3\}$, as follows.

\begin{minipage}[c]{0.42\linewidth}
For $x\in V(G):$
\vspace{-0.2cm}
$$p(x)=
\begin{cases}
1, \mbox{ if } x \in G_0 \cup F_{a} \cup F_{b}; \\
2, \mbox{ if } x = a \mbox{ or } x=b; \\
3, \mbox{ otherwise. }
\end{cases}$$
\end{minipage}
and
\hfill
\begin{minipage}[c]{0.6\linewidth}
\vspace{0.2cm}
$$p(\overline{x})=
\begin{cases}
1, \mbox{ if } x = a \mbox{ or } x=b;\\
2, \mbox{ if } x \not \in (N_{G}[a] \cap N_{G} [b]);\\
3, \mbox{ if } x\in N_{G}(a) \cap N_{G}(b).
\end{cases}$$
\end{minipage}
\vspace{0.2cm}
We show that $p$ is a $R_{3,4}$-role assignment.

The vertices with role~1 are those belonging to $(G_0 \cup F_{a} \cup F_{b})$ or the vertices $\overline{a}$ and $\overline{b}$. Let $x\in G_0$, we have that $p(N(x))=\{p(\overline{x})\}=\{2\}$. Let $x\in F_{a}$, then $N(x)= \{\overline{x},a\}$ and $p(N(x))=\{2\}$. We proceed similarly for the vertices of $F_{b}$. Let $\overline{x}$, with $x\in \{a, b\}$, we have that $N(\overline{x}) \subseteq \{x\} \cup V(G)- (V(G)[a] \cap N_{G}[b])$, so $p(N(\overline{x}))=\{2\}.$

The vertices with role~2 are $a$, $b$ or those belonging to $\overline{V(G)-(N_{G}[a] \cap N_{G}[b])}$. For the vertex $a$, we have $p(\overline{a})=1$, $p(b)=2$ and $p(c)=3$. Similarly, it goes to the vertex $b$. Let $\overline{x}$, with $x \in G_0 \cup F_a\cup F_b$, we have $p(x)=1$. By the hypothesis $|F_{a} \cup F_{b}| \geq 2$, so $2 \in p(N(\overline{x}))$, since $p(\overline{F_a \cup F_b})=\{2\}$. Finally, since $p(\overline{c})=3$, we have that $p(N(\overline{x}))=\{1,2,3\}$. Let $\overline{x}$, with $x\not \in N_{G}[a]\cap N_{G}[b]$, $x\not \in G_0 \cup F_a \cup F_b$. We can assume that $x \not \in N_{G}(a)$. Thus, it follows from the fact that $p(x)=3$, $p(\overline{a})=1$ and $p(\overline{F_{a} \cup F_{b}})=\{2 \}$, that $p(N(\overline{x}))=\{1,2,3\}$.

The vertices with role~3 are those belonging to $V(G)-(\{a,b\} \cup G_0\cup F_{a}\cup F_{b})$ or $\overline{N_{G }(a) \cap N_{G}(b)}$. Let $x\in V(G)-(\{a,b\} \cup G_0 \cup F_{a} \cup F_{b})$. If $x \in N_{G}(a)\cap N_{G}(b)$, then $p(\overline{x})=3$ and $p(a)=2$. Otherwise, $x \not \in N_{G}(a)\cap N_{G}(b) $, so $p(\overline{x})= 2$. Since $x$ is not an isolated vertex, there is $y\in V(G)$, such that $xy\in E(G)$. Clearly, $y \not \in G_0 \cup F_a \cup F_b$. As $x \not \in F_a \cup F_b$, we can choose $y \neq a,b$. Therefore, $p(y)=3$. Let $\overline{x}$, with $x\in N_{G}(a)\cap N_{G}(b)$, we have $p(x)=3$ and $p(\overline{F_a \cup F_b})=\{2\}$, so $p(N(\overline{x}))=\{2,3\}.$\end{proof}

In the following lemma, we conclude that a sufficient condition for the complementary prisms of $G$ to have a $3$-role assignment is the existence of an isolated vertex and a clique of order~3. Observe that the graph $K_1 \cup K_n$, with $n\geq 3$, is the complementary graph of $K_{1,n}$ that has a $R_{3,2}$-role assignment by Lemma~\ref{lemma:k1m}.

\begin{lemma} \label{lemma:nonsplit_noleaf_comisolated} Let $G$ be a graph with isolated vertices, such that $G \not \simeq K_1 \cup K_n$ with $n\geq 3$. If $G$ has a clique of order~3, then the complementary prism of $G$ has a $R_{3,4}$-role assignment. \end{lemma}

\begin{proof} We consider $\{a,b,c\}$ a clique of $G$ and $d\in G_0$. By Lemma~\ref{lemma:connection_semiisolated_naobipartite_withtriangle_2leaves}, this clique has at most one leaf neighbor. If so, we can assume that the leaf is a neighbor of $c$, and thus $a$ and $b$ have no leaf neighbors. Since $d\not \in N_{G}[a] \cap N_{G}[b]$ for all $x\neq a,b,d$, we have that $xd \not \in E(G)$. By Lemma~\ref{lemma:connected_semiisolated_notbipartite_withtriangle_with1ornoverticalleaf}, we can assume that $V(G)=(N_{G}[a]\cap N_{G}[b]) \cup \{d\}$ and therefore $G$ has no leaf. It follows from the fact that $G$ is not isomorphic to $K_1 \cup K_n$ and $d \in G_0$ that there are $u,v \in N_{G}[a]\cap N_{G}[b]$, such that $uv \not \in E(G)$. Clearly, $u$ and $v$ are different from $a$ and $b$. On the other hand, $\{a,b,u\}$ is a clique, so $u$ and $a$ have no leaf neighbors, for all $x \neq u,a,d$ we have $xd \not \in E(G)$ and we have $v\not \in N_{G}[u] \cap N_{G}[a]$ with $dv \not \in E(G)$. Therefore, the result follows from the Lemma~\ref{lemma:connected_semiisolated_notbipartite_withtriangle_with1ornoverticalleaf}. \end{proof}

From the previous lemma, in addition to assuming that $G$ has a clique of order~3, we can also assume that $G$ has no isolated vertices. Next, we display lemmas contemplating the cases of graphs that have leaves and do not satisfy the condition of Lemma~\ref{lemma:connected_semiisolated_notbipartite_withtriangle_with1ornoverticalleaf}.

\begin{lemma}\label{lemma:nobipartite_Noleaf_semiisolated_noleafneighbor} Let $G$ be a graph with a single leaf $f$ and no isolated vertices. Let $\{a,b,c\}$ be a clique of $G$, such that, $N_{G}(f)=\{c\}$. If $V(G) \neq N_{G}[a] \cup \{f\}$ and for all $y \not \in N_{G}[a] \cap N_{G}[b]$, we have that $cy\in E(G)$, then the complementary prisms of $G$ has a $R_{3,4}$-role assignment. \end{lemma}

\begin{proof} We denote by $C=\{x\in V(G) \mid N_{G}[x]=V(G)-\{f\}\}$. Observe that, for all $x \in C$, $\overline{x}$ is a leaf in $\overline{G}$. Also, $C$ is a possibly empty clique and $a,f \not \in C$. We define $p:V(G\overline{G})\rightarrow \{1,2,3\}$, as follows. For $x\in V(G):$

\begin{minipage}[c]{0.4\linewidth}
\vspace{-0.2cm}
$$p(x)=
\begin{cases}
1, \mbox{ if } x=a \mbox{ or } x=f;\\
2, \mbox{ if } x\in N_{G}(a); \\
3, \mbox{ otherwise. }
\end{cases}$$
\end{minipage}
and
\hfill
\begin{minipage}[c]{0.7\linewidth}
$$p(\overline{x})=
\begin{cases}
1, \mbox{ if } x\in C; \\
2, \mbox{ if } x = a \mbox{ or } x=f;\\
3, \mbox{ otherwise. }
\end{cases}$$
\end{minipage}

\vspace{0.2cm}
We consider $u\in V(G)-(N_{G}[a]\cup \{f\})$. We note that $u\not \in C$, $p(u)=3$ and $p(\overline{u})=3$. In view of Lemma~\ref{lemma:nonsplit_noleaf_comisolated}, we emphasize that, as there is a clique in $G$, we have $G \not \simeq K_{1, n}$ and thus $\overline{G} \not \simeq K_1 \cup K_n $. On the other hand, $\{\overline{a}, \overline{u}, \overline{f}\}$ is a clique in $\overline{G}$. If $c$ is a universal vertex, then $\overline{c}$ is an isolated vertex in $\overline{G}$ and the result follows from Lemma~\ref{lemma:nonsplit_noleaf_comisolated}. Therefore, suppose that $c$ is not a universal vertex and we show that $p$ is a $R_{3,4}$-role assignment.

The vertices with role~1 are $a$, $f$ or those belonging to $\overline{C}$. For the vertex $a$, we have that $p(N(a)) = \{2\}$. For the vertex $f$, we have that $N(f)= \{\overline{f}, c\}$ and $p(N(f))=\{2\}$. Let $\overline{x}$, with $x\in C$, we have that $N(\overline{x})=\{x, \overline{f}\}$ and $C \subseteq N_{G} (a)$, so $p(N(\overline{x}))=\{2\}$.

The vertices with role~2 are those belonging to $N_{G}(a)$ or the vertices $\overline{a}$ and $\overline{f}$. Let $x \in N_{G}(a)$, we have $p(a)=1$. If $x \in C$, then we have $p(b)=2$ and $p(u)=3$. Otherwise, $p(\overline{x})=3$ and it remains to show a neighbor of $x$ with role~2. If $x\in N_{G}(b)$, then $p(b)=2$. Otherwise, by the hypotheses, $cx\in E(G)$ and $2\in p(N(x))$, therefore, $p(N(x))=\{1,2,3\}$. For the vertex $\overline{a}$, we have $p(a)=1$, $p(\overline{f})=2$ and $p(\overline{u})=3$. For the vertex $\overline{f}$, we have $p(f)=1$, $p(\overline{a})=2$ and $p(\overline{u})=3$.

The vertices with role~3 are those belonging to $V(G)-(N_{G}[a]\cup \{f\})$ or to $\overline{V(G)-(\{a,f \}\cup C)}$. Let $x\in V(G)-(N_{G}[a]\cup \{f\})$. So, $x \not \in C$ and $p(\overline{x})=3$. By the hypothesis, $cx \in E(G)$ and $2\in p(N(x))$. Let $\overline{x}$, with $x\in V(G)-(\{a,f\}\cup C)$. If $x \in N_{G}(a)$, then $p(x)=2$. We note that the vertex $c$ is not universal. As $x \not \in C$, there is $y\in V(G), y\neq f$, $xy\not \in E(G)$. Clearly, $y\neq a$ and $y \not \in C$. So, $p(\overline{y})=3$ and $3\in p(N(\overline{x}))$. If $x\not \in N_{G}(a),$ then $p(x)=3$ and $p(\overline{a})=2$. \end{proof}

 The next lemma contemplates similar conditions to the previous lemma.

\begin{lemma} \label{lemma:nobipartite_with-vertex_semiisolated_leaf_withclick} 
Let $G$ be a graph with a single leaf $f$ and no isolated vertices. 
Let $\{a,b,d\}$ be a clique of $G$, such that $V(G)=(N_{G}[a] \cap N_{G}[b]) \cup \{f \}$ and $N_{G}(f)=\{d\}$. If $d$ is not a universal vertex, then the complementary prism of $G$ has a $R_{3,4}$-role assignment. \end{lemma}

\begin{proof} We define $p:V(G\overline{G})\rightarrow \{1,2,3\}$, as follows. For $x\in V(G):$

\begin{minipage}[c]{0.4\linewidth}
\vspace{-0.2cm}
$$p(x)=
\begin{cases}
1, \mbox{ if } x= f;\\
2, \mbox{ if } x=a \mbox{ or } x=d; \\
3, \mbox{ otherwise. }
\end{cases}$$
\end{minipage}
and
\hfill
\begin{minipage}[c]{0.5\linewidth}
$$p(\overline{x})=
\begin{cases}
1, \mbox{ if } x=a \mbox{ or } x=d; \\
2, \mbox{ if } x = f \mbox{ or } x \not \in N_{G}[d];\\
3, \mbox{ if } x \in N_{G}(d)-\{a,f\}.
\end{cases}$$
\end{minipage}
\vspace{0.2cm}
One can easily see $p$ is a $R_{3,4}$-role assignment.\end{proof}




The next lemma deals with the remaining cases for the desired result when $G$ has leaves.

\begin{lemma}\label{lemma:nobipartite_withfolhas_independentset1} Let $G$ be a graph. Let $a,b,c,d,f$ be vertices of $G$ such that $\{a,b,c\}$ is a clique, $f$ is a leaf with $N_{G}(f) =\{d\}$, $N_{G}(d)=\{a,f\}$ and $V(G)=N_{G}[a]\cup \{f\}$. The complementary prism of $G$ has a $R_{3,4}$-role assignment. \end{lemma}

\begin{proof} We define $p:V(G\overline{G})\rightarrow \{1,2,3\}$, as follows. For $x\in V(G):$

\begin{minipage}[c]{0.5\linewidth}
$$p(x)=
\begin{cases}
1, \mbox{ if } x=c;\\
2, \mbox{ if } x\in N_{G}(c); \\
3, \mbox{ otherwise. }
\end{cases}$$
\end{minipage}
and
\hfill
\begin{minipage}[c]{0.5\linewidth}
$$p(\overline{x})=
\begin{cases}
1, \mbox{ if } x=a; \\
2, \mbox{ if } x = c \mbox{ or } x=f;\\
3, \mbox{ otherwise. }
\end{cases}$$
\end{minipage}

Observe that $p(d)=3$, $p(\overline{d})=3$. With this in mind, we can easily verify that $p$ is a $R_{3,4}$-role assignment. \end{proof}




The following theorem shows that the complementary prisms of a non-bipartite graph with leaves, such that the complementary graph is non-bipartite, have a $3$-role assignment.

\begin{theorem} \label{theorem:nobipartite_withverticleaf} Let $G$ be a non-bipartite graph with at least one leaf and such that $\overline{G}$ is non-bipartite. The complementary prism of $G$ has a $3$-role assignment. \end{theorem}

\begin{proof} By Theorem~\ref{theorem:nobipartite_withverticeisolated_withoutclick}, we can assume that $G$ and $\overline{G}$ have a clique of order~3. As $G$ and $\overline{G}$ are non-bipartite, neither is isomorphic to $K_{1,n}$, nor to $\overline{K_{1,n}}\simeq K_1 \cup K_n$. So, by Lemma~\ref{lemma:nonsplit_noleaf_comisolated}, we can assume that neither $G$ nor $\overline{G}$ has isolated vertices, that is, $G$ has neither isolated vertices nor universal vertices.

We consider $\{a,b,c\}$ a clique of $G$. By Lemma~\ref{lemma:connection_semiisolated_naobipartite_withtriangle_2leaves}, we can assume that $a$, $b$ have no leaf neighbors. By Lemma~\ref{lemma:connected_semiisolated_notbipartite_withtriangle_with1ornoverticalleaf}, we can assume that there is $x\neq a,b$, such that for all $y \in V(G)-(N_{G}[a] \cap N_{G}[b])$, $y\neq x$, $xy \in E(G)$. Let $f\in G_1$ be such that $N_{G}(f)=\{d\}$. As $f \not \in N_{G}[a]\cap N_{G}[b]$, we have $x=f$ or $x=d$. Observe that if $x=f$, then $V(G)-(N_{G}[a] \cap N_{G}[b])\subseteq \{d,f\}$ and the condition is also satisfied for $x=d$. Therefore, $dy\in E(G)$ for all $y \in V(G)-( N_{G}[a] \cap N_{G}[b])$, $y\neq d$, that is, $V(G)=( N_{G}[a] \cap N_{G}[b]) \cup N_{G} [d]$.

First, suppose that $c=d$. Observe that $N_{G}(G_1)=\{d\}$ and by Lemma~\ref{lemma:connection_semiisolated_naobipartite_withtriangle_2leaves}, we can assume that $G_1=\{f\}$. Remember that $G$ does not have a universal vertex. By the Lemma~\ref{lemma:nobipartite_Noleaf_semiisolated_noleafneighbor}, we can assume that $V(G)=N_{G}[a]\cup \{f\} =N_{G}[b] \cup \{f\}$ , that is, $V(G)=(N_{G}[a]\cap N_{G}[b]) \cup \{f\}$. The result follows from the Lemma~\ref{lemma:nobipartite_with-vertex_semiisolated_leaf_withclick}. Therefore, $d$ is not part of any clique of order~3, this is $N_{G}(d)$ is an independent set. Remember that $V(G)=(N_{G}[a] \cap N_{G}[b])\cup N_{G}(d).$

Let $u\in N_{G}(d)$. If $u \not \in G_1$, then, by Lemma~\ref{lemma:nobipartite_semi-isolated_clique_at least1verticalleaf}, we can assume that $V(G)=N_{G}[u]\cup \{f\}$. In particular, we can assume that $u=a$ and $b,c \in N_{G}(u)$. On the other hand, as $N_{G}(d)$ is an independent set, $N_{G}(d)=\{u,f\}$ and the result follows from the Lemma~\ref{lemma:nobipartite_withfolhas_independentset1} .

Therefore, $N_{G}(d)=G_1-\{d\}$ and the connected component containing $f$ is isomorphic to $K_{1,m}$ for some $m\geq 1$. By the above, $G \simeq G' \cup K_{1,m}$ where $G'=N_{G}[a]\cap N_{G}[b]$ and $m=|N_{G}( d)|$. If $G' \simeq K_n$, for some $n \geq 3$, then the conclusion follows from Lemma~\ref{lemma:k1m,k1}. Otherwise, there are $u,v \in V(G')$, such that $uv\not \in E(G)$. On the other hand, $\{a,b,u\}$ forms a clique. Again, by Lemma~\ref{lemma:connected_semiisolated_notbipartite_withtriangle_with1ornoverticalleaf}, we can assume that $dy\in E(G)$ for all $y \not \in N_{G}[a] \cap N_{G}[u]$ , with $y \neq d$. However, $v \not \in N_{G}[a]\cup N_{G}[u]$ and $dv \not \in E(G)$, a contradiction. \end{proof}

To prove the final theorem, we need a last lemma, which deals with a graph $G$ with no universal vertex, no isolated vertices, and no leaves.

\begin{lemma}\label{lemma:naobipartite_naoisomorphC5_N>=6_leafless} Let $G$ be a graph with no universal vertex and no isolated vertices such that neither $G$ nor $\overline{G}$ has leaves. If $\{a,b,c\}$ is a clique of $G$ and $f,d \in V(G)-\{a,b\}$, $f\neq d$, satisfies the following conditions:

\begin{enumerate}
\item for all $y\not \in N_{G}[a]\cap N_{G}[b],$ with $y\neq d$ we have $dy \in E(G)$;
\item $af\not \in E(G)$;
\item $(N_{G}(a) \cup N_{G}(f)) \cap N_{G}(d)\neq \emptyset$, \end{enumerate}
then the complementary prisms of $G$ has a $R_{3,4}$-role assignment. \end{lemma}

\begin{proof} We define $p:V(G\overline{G})\rightarrow \{1,2,3\}$, as follows. For $x\in V(G):$

\begin{minipage}[c]{0.5\linewidth}
$$p(x)=
\begin{cases}
1, \mbox{ if } x=a \mbox{ or } x=f;\\
2, \mbox{ if } x \in N_{G}(a)\cup N_{G}(f); \\
3, \mbox{ otherwise. }
\end{cases}$$
\end{minipage}
and
\hfill
\begin{minipage}[c]{0.5\linewidth}
$$p(\overline{x})=
\begin{cases}
2, \mbox{ if } x=a \mbox{ or } x=f;\\
3, \mbox{ otherwise. }
\end{cases}$$
\end{minipage}
\vspace{0.2cm}
As $f \not \in N_{G}[a]$ we have, by Condition~1, that $df\in E(G)$. So, $p(d)=2$ and $p(\overline{d})=3$. We consider $u \in (N_{G}(a)\cup N_{G}(f))\cap N_{G}(d)$ and observe that, $p(u)=2$ and $p(\overline{u})=3$. In light of this information, we can easily conclude that  $p$ is a $R_{3,4}$-role assignment. \end{proof}




We end with Theorem~\ref{theorem_nobipartite_withoutleaf}, by showing that the complementary prisms of a non-bipartite graph, whose complementary graph is non-bipartite, has a $3$-role assignment.

\begin{theorem}\label{theorem_nobipartite_withoutleaf} Let $G$ be a graph, such that neither $G$ nor $\overline{G}$ are bipartite graphs. The complementary prism of $G$ has a $3$-role assignment. \end{theorem}

\begin{proof} We can assume from Theorems~\ref{theorem:nobipartite_withverticeisolated_withoutclick}, ~\ref{theorem:nobipartite_withverticleaf} and from Lemma~\ref{lemma:nonsplit_noleaf_comisolated} that $G$ and $\overline{G}$ have a clique of order~3, they have no leaves and no isolated vertex.

Let $\{a,b,c\}$ be a clique of $G$. By Lemma~\ref{lemma:connected_semiisolated_notbipartite_withtriangle_with1ornoverticalleaf}, we can assume that there exists $d \neq a,b$, such that, for all $y \not \in N_{G}[a] \cap N_{G}[b ]$, $y\neq d$, we have that $dy\in E(G)$, this is $V(G)=(N_{G}[a] \cap N_{G}[b]) \cup N_{G}[d]$. As $\overline{a}$ is neither a leaf nor an isolated vertex, there exists a vertex $f \neq d$, such that $af\not \in E(G)$. By Lemma~\ref{lemma:naobipartite_naoisomorphC5_N>=6_leafless}, we can assume that $(N_{G}(a)\cup N_{G}(f)) \cap N_{G}(d)=\emptyset$, that is, $a$ and $d$ have not neighbors in common. If $ad \in E(G)$, then $\{a,d\}$ is a maximal clique. We show that this clique satisfies the conditions of Lemma~\ref{lemma:naobipartido_naosplit_maximum clique}. Condition~1 follows from the fact that $G$ has no leaves. Condition~2 follows from the fact that $\overline{G}$ has no isolated vertex or leaf and from the property of $d$. Condition~3 follows from the fact that $G$ has no isolated vertices, no leaf, and that $a$ and $d$ have no neighbors in common. We note that, in fact, we use that $V(G)=N_{G}[a]\cup N_{G}[d]$ and $N_{G}(a)\cap N_{G}(d)= \emptyset$. We note that $\overline{a}$ and $\overline{d}$ also satisfy these conditions in $\overline{G}$. Therefore, if $ad\not \in E(G)$, then the maximal clique $\{\overline{a}, \overline{d}\}$ satisfies the conditions of Lemma~\ref{lemma:naobipartido_naosplit_maximum clique} in $\overline{G}$. Thus, the complementary prism of $G$ has a $3$-role assignment. \end{proof}

\section{Conditions to have a 3-role assignment in complementary prisms}
\label{subsection_characterization_prism}

We have, as a final result, the characterization of complementary prisms with a $3$-role assignment, in which we conclude that there are only some exceptions of graphs that do not have such an assignment.

\begin{theorem} \label{nonbipartite_final_theorem} Let $G$ be a graph. The complementary prisms of $G$ has a $3$-role assignment if and only if $G$ and $\overline{G}$ are not isomorphic to one of the following graphs:
\begin{enumerate}
\item $K_n$, with $n\geq 2$;
\item the graph $G$ has isolated vertices and every non-trivial connected component is isomorphic to $K_{n,m}$, for some $n,m\geq 1$, with the exception of the graph $ K_1 \cup K_{1,m}$;
\item $K_2^t$, with $t\geq 4$;
\item $\bigcup^t_{i=1} K_{1,m_i}$, with $t\geq 3$, $m_1 \geq 2$, $m_i \geq 1$, $i=2, \ldots, t$;
\item $K_{1,m_1}\cup K_{1,m_2}$, with $m_1,m_2 \geq 2$. \end{enumerate} \end{theorem}

\begin{proof} Follows from Theorems~\ref{general_bipartite_theorem} and~\ref{theorem_nobipartite_withoutleaf} 
\end{proof}

Observe that all complementary prisms that do not have a $3$-role assignments arise from disconnected bipartite graphs.

We conclude with Theorem~\ref{theorem_linear_prisms_complexity} that the computational complexity of deciding whether a complementary prism has a $3$-role assignment is polynomial. To be more specific, we have shown in \citep{mesquita2022atribuiccao} that such a problem can be solved in linear time.

\begin{theorem}\label{theorem_linear_prisms_complexity} The $3$-\textsc{Role Assignment} problem for complementary prisms can be decided in polynomial time. \end{theorem}

\begin{proof} Let $G=(V(G), E(G))$ be a graph with $|V(G)|=n$ and $|E(G)|=m$. An algorithm that solves the problem by using Theorem~\ref{nonbipartite_final_theorem} returns that $G\overline{G}$ does not have a $3$-role assignment when $G$ or $\overline{G}$ are not isomorphic to graphs of Conditions $1$ to $5$. In this way, we show how to check each of the five conditions in polynomial time.

First, Condition $1$ specifies to check whether $G \simeq K_n$. Thus, we identify if the degree of every vertex of $G$ is exactly $n-1$, which can be done in time $\mathcal{O}(n+m) $. Observe that, in Conditions $2$ to $5$, it is necessary to check whether $G$ is the disjoint union of complete bipartite graphs. To this end, detecting the connected components of $G$, as well as the bipartite partitions, can be done in $\mathcal{O}(n+m)$ time~\citep{kleinberg2006algorithm}. 
After obtaining the bipartitions of each component, Conditions $2$ to $5$ can be checked only by examining the degree of each vertex, which runs in $\mathcal{O}(n+m)$ time. As the degree of each vertex is an integer between $0$ and $n-1$, obtaining the degree sequence $d=(d_1, \ldots, d_n)$, with $d_1 \leq d_2 \leq \ldots, \leq d_n $, can also be performed in $\mathcal{O}(n+m)$ time ~\citep{cormen2009introduction}. The complement graph $\overline{G}$ of be computed in $\mathcal{O}(n^2)$.

If $G$ and $\overline{G}$ do not satisfy any of the Conditions $1$ to $5$ we simply return that $G\overline{G}$ has a $3$-role assignment. This whole procedure, therefore, requires running time of order $\mathcal{O}(n^2)$. We conclude that $3$-\textsc{Role Assignment} for complementary prisms can be decided in polynomial time. \end{proof} 

\bibliographystyle{abbrvnat}
\bibliography{prisma_polinimial-dmtcs}
\label{sec:biblio}

\end{document}